\documentclass{elsarticle}

\usepackage{graphicx}
\usepackage{amsfonts}
\usepackage{amssymb}
\usepackage{amsmath}
\usepackage{amsthm}
\usepackage{multirow}
\usepackage[english]{babel}

\newcommand{\tr}{\operatorname{tr}}
\newcommand{\E}{\mathsf{E}}
\newcommand{\VAR}{\mathsf{VAR}}
\newcommand{\COV}{\mathsf{COV}}
\newcommand{\Prob}{\mathsf{P}}

\newtheorem{thm}{Theorem}

\title{Measures of Variability for Bayesian Network Graphical Structures}

\author{Marco Scutari}
\address{Department of Statistical Sciences, University of Padova}
\begin{document}

\begin{abstract}
The structure of a Bayesian network includes a great deal of information about
the probability distribution of the data, which is uniquely identified given 
some general distributional assumptions. Therefore it's important to study its
variability, which can be used to compare the performance of different learning
algorithms and to measure the strength of any arbitrary subset of arcs.

In this paper we will introduce some descriptive statistics and the
corresponding parametric and Monte Carlo tests on the undirected graph
underlying the structure of a Bayesian network, modeled as a multivariate 
Bernoulli random variable. A simple numeric example and the comparison
of the performance of some structure learning algorithm on small samples
will then illustrate their use.
\end{abstract}


\maketitle

\section{Introduction}

In recent years Bayesian networks have been successfully applied in several
different disciplines, including medicine, biology and epidemiology (see for
example \citet{gene} and \citet{holmes}). This has been made possible by the
rapid evolution of structure learning algorithms, from constraint-based ones
(such as PC \citep{spirtes}, Grow-Shrink \citep{mphd}, IAMB \citep{iamb} and
its variants \citep{fastiamb}) to score-based (such as TABU search \citep{russel},
Greedy Equivalent Search \citep{ges} and genetic algorithms \citep{larranaga}) 
and hybrid ones (such as Max-Min Hill Climbing \citep{mmhc}).

The main goal in the development of these algorithms has been the reduction of
the number of either independence tests or score comparisons needed to learn
the structure of the Bayesian network. Their correctness has been proved assuming
either very large sample sizes in relation to the number of variables (in the
case of Greed Equivalent Search) or the absence of both false positives and
false negatives (in the case of Grow-Shrink and IAMB). In most cases the
characteristics of the learned networks were studied using a small number of
reference data sets \citep{bnr} as benchmarks, and differences from the true
structure measured with purely descriptive measures such as Hamming distance
\citep{graphs}.

This approach to model evaluation is not possible for real world data sets, as
the true structure of their probability distribution is not known.
An alternative is provided by the use of either parametric or nonparametric
bootstrap \citep{efron}. By applying a learning algorithm to a sufficiently
large number of bootstrap samples it is possible to obtain the empirical
probability of any feature of the resulting network \citep{friedman}, such as
the structure of the Markov Blanket of a particular node. The fundamental
limit in the interpretation of the results is that the ``reasonable'' level
of confidence for thresholding depends on the data.

In this paper we propose a modified bootstrap-based approach for the inference
on the structure of a Bayesian network. The undirected graph underlying the
network structure is modeled as a multivariate Bernoulli random variable in
which each component is associated with an arc. This assumption allows the
derivation of both exact and asymptotic measures of the variability of the
network structure or any of its parts.

\section{Bayesian networks and bootstrap}
\label{sec:bn}

Bayesian networks are graphical models where nodes represent random variables
(the two terms are used interchangeably in this article) and arcs represent
probabilistic dependencies between them \citep{korb}.

The graphical structure $\mathcal{G} = (\mathbf{V}, A)$ of a Bayesian network is
a \textit{directed acyclic graph} (DAG) which defines a factorization of the
joint probability distribution of $\mathbf{V} = \{X_1, X_2, \ldots, X_v\}$,
often called the \textit{global probability distribution}, into a set of
\textit{local probability distributions}, one for each variable. The form of 
the factorization is given by the \textit{Markov property} of Bayesian networks,
which states that every random variable $X_i$ directly depends only on its
parents $\Pi_{X_i}$.

Therefore it is important to define confidence and variability measures for 
specific features in the network structure, such as the presence of specific 
configurations of arcs. In particular a measure of variability for the network
structure as a whole has many applications both as an indicator of goodness of
fit for a particular Bayesian network and as a criterion to evaluate the 
performance of a learning algorithm.

Confidence measures have been developed by \citet{friedman} using bootstrap 
simulation, and later modified by \citet{imoto} to estimate the marginal confidence
in the presence of an arc (called \textit{edge intensity}, and also known as
\textit{arc strength}) and its direction. This approach can be summarized as 
follows:
\begin{enumerate}
  \item For $b = 1, 2, \ldots, m$
    \begin{enumerate}
      \item re-sample a new data set $\mathbf{D^*_b}$ from the original data
        $\mathbf{D}$ using either parametric or nonparametric bootstrap.
      \item learn a Bayesian network $\mathcal{G}_b$ from $\mathbf{D^*_b}$.
    \end{enumerate}
  \item Estimate the confidence in each feature $f$ of interest as
    $\hat\Prob(f) = (1/m) \sum_{b=1}^m f(\mathcal{G}_b)$.
\end{enumerate}

However, the empirical probabilities $\hat\Prob(f)$ are difficult to evaluate,
because the distribution of $\mathcal{G}$ in the space of DAGs is unknown and 
because the confidence threshold value depends on the data.

\section{The multivariate Bernoulli distribution}
\label{sec:mvb}

Let $B_1, B_2, \ldots, B_k$, $k \in \mathbb{N}$ be Bernoulli random variables
with marginal probability of success $p_1, p_2, \ldots, p_k$, that is
$B_i \sim Ber(p_i)$, $i = 1, \ldots, k$. Then the distribution of the random
vector $\mathbf{B} = [B_1, B_2, \ldots, B_k]^T$ over the joint probability
space of $B_1, B_2, \ldots, B_k$ is a \textit{multivariate Bernoulli random
variable} \citep{krummenauer}, denoted as $Ber_k(\mathbf{p})$. Its probability
function is uniquely identified by the parameter collection $\mathbf{p} = 
\left\{ p_I : I \subseteq \{1, \ldots, k \}, I \neq \varnothing \right\}$, which
represents the \textit{dependence structure} among the marginal distributions in
terms of simultaneous successes for every non-empty subset $I$ of elements of
$\mathbf{B}$.

However, several useful results depend only on the first and second order moments
of $\mathbf{B}$
\begin{align}
  \E(B_i) &= p_i \\
  \label{eq:var} \VAR(B_i) &= \E(B_i^2) - \E(B_i)^2 = p_i - p_i^2 \\
  \label{eq:cov} \COV(B_i, B_j) &= \E(B_i B_j) - \E(B_i) \E(B_j)= p_{ij} - p_i p_j
\end{align}
and the reduced parameter collection 
$\mathbf{\tilde{p}} = \left\{ p_{ij} : i,j = 1, \ldots, k \right\}$,
which can be used as an approximation of $\mathbf{p}$ in the generation random 
multivariate Bernoulli vectors in \citet{mvbsim}.

\subsection{Uncorrelation and independence}

We will first consider a simple result that links covariance and independence of
two univariate Bernoulli variables.
\begin{thm}
\label{thm:univindep}
  Let $B_i$ and $B_j$ be two Bernoulli random variables. Then $B_i$ and $B_j$
  are independent if and only if they are uncorrelated.
\end{thm}
\begin{proof}
  If $B_i$ and $B_j$ are independent, then by definition $\COV(B_i, B_j) = 0$.
  If on the other hand we have that $\COV(B_i, B_j) = 0$, then $p_{ij} = p_i p_j$
  which completes the proof.
\end{proof}

This theorem can be extended to multivariate Bernoulli random variables as follows.

\begin{thm}
  Let $\mathbf{B} = [B_1, B_2, \ldots, B_k]^T$ and $\mathbf{C} = [C_1, C_2, \ldots, C_l]^T$,
  $k,l \in \mathbb{N}$ be two multivariate Bernoulli random variables. Then $\mathbf{B}$ and
  $\mathbf{C}$ are independent if and only if $\COV(\mathbf{B}, \mathbf{C}) = \mathbf{O}$,
  where $\mathbf{O}$ is the zero matrix.
\end{thm}
\begin{proof}
  If $\mathbf{B}$ is independent from $\mathbf{C}$, then by definition every pair
  $(B_i, C_j)$, $i = 1, \ldots, k$, $j = 1, \ldots, l$ is independent. Therefore
  $\COV(\mathbf{B}, \mathbf{C}) = [c_{ij}] = \mathbf{O}$. 

  If conversely $\COV(\mathbf{B}, \mathbf{C}) = \mathbf{O}$, every pair $(B_i, C_j)$ 
  is independent as $c_{ij} = 0$ implies $p_{ij} = p_i p_j$. This in turn implies
  the independence of the random vectors $\mathbf{B}$ and $\mathbf{C}$, as their
  sigma-algebras $\sigma(\mathbf{B}) = \sigma(B_1) \times \ldots \times \sigma(B_k)$
  and $\sigma(\mathbf{C}) = \sigma(C_1) \times \ldots \times \sigma(C_l)$
  are functions of the sigma-algebras induced by the two sets of independent random
  variables $B_1, B_2, \ldots, B_k$ and $C_1, C_2, \ldots, C_l$.
\end{proof}

The correspondence between uncorrelation and independence is identical to the
analogous property of the multivariate Gaussian distribution \citep{ash}, and
is closely related to the strong normality defined for orthogonal second order
random variables in \citet{loeve}. It can also be applied to disjoint subsets
of components of a single multivariate Bernoulli variable, which are also
distributed as multivariate Bernoulli random variables.

\begin{thm}
  Let $\mathbf{B} = [B_1, B_2, \ldots, B_k]^T$ be a multivariate Bernoulli random
  variable; then every random vector $\mathbf{B^*} = [B_{i_1}, B_{i_2}, \ldots,
  B_{i_l}]^T$, $\left\{i_1, i_2, \ldots, i_l\right\} \subseteq \left\{1, 2, \ldots,
  k\right\}$ is a multivariate Bernoulli random variable.
\end{thm}
\begin{proof}
  The marginal components of $\mathbf{B^*}$ are Bernoulli random variables, because
  $\mathbf{B}$ is multivariate Bernoulli. The new dependency structure is defined as
  $\mathbf{p^*} = \left\{ p_{I^*} : I^* \subseteq \{i_1, \ldots, i_l \}, I^* \neq \varnothing \right\}$,
  and uniquely identifies the probability distribution of $\mathbf{B^*}$.
\end{proof}

\subsection{Properties of the covariance matrix}

The covariance matrix $\Sigma = [\sigma_{ij}]$, $i,j = 1, \ldots, k$ associated
with a multivariate Bernoulli random vector has several interesting numerical
properties. Due to the form of the central second order moments defined in
formulas \ref{eq:var} and \ref{eq:cov}, the diagonal elements $\sigma_{ii}$ are
bounded in the interval $[0, 1/4]$. The maximum is attained for $p_i = 1/2$, and
the minimum for both $p_i = 0$ and $p_i = 1$. For the Cauchy-Schwarz theorem 
\citep{ash} then $|\sigma_{ij}| \in \left[0, 1/4\right]$.

The eigenvalues $\lambda_1, \lambda_2, \ldots, \lambda_k$ of $\Sigma$ are
similarly bounded, as shown in the following theorem.

\begin{thm}
  Let $\mathbf{B} = [B_1, B_2, \ldots, B_k]^T$ be a multivariate Bernoulli random
  variable, and let $\Sigma = [\sigma_{ij}]$, $i,j = 1, \ldots, k$ be its covariance
  matrix. Let $\lambda_i$, $i = 1, \ldots, k$ be the eigenvalues of $\Sigma$. Then
  $0 \leqslant \sum_{i=1}^k \lambda_i \leqslant k/4$ and
  $0 \leqslant \lambda_i \leqslant k/4$.
\end{thm}
\begin{proof}
  Since $\Sigma$ is a real, symmetric, non-negative definite matrix, the eigenvalues
  $\lambda_i$ are non-negative real numbers \citep{salce}; this proves the lower bound
  in both inequalities.

  The upper bound in the first inequality holds because
  \begin{equation}
     \sum_{i=1}^k \lambda_i = \sum_{i=1}^k \sigma_{ii} \leqslant
     \max_{\left\{\sigma_{ii}\right\}} \sum_{i=1}^k \sigma_{ii} =
     \sum_{i=1}^k \max \sigma_{ii} = \frac{k}{4},
  \end{equation}
  and this in turn implies $\lambda_i \leqslant \sum_{i=1}^k \lambda_i \leqslant k/4$,
  which completes the proof.
\end{proof}

These bounds define a convex set in $\mathbb{R}^k$, defined by the family
\begin{equation}
  \mathcal{D} = \left\{ \Delta^{k-1}(c) : c \in \left[ 0, \frac{k}{4} \right]\right\}
\end{equation}
where $\Delta^{k-1}(c)$ is the non-standard $k-1$ simplex
\begin{equation}
  \Delta^{k-1}(c) = \left\{ (\lambda_1, \ldots, \lambda_k) \in \mathbb{R}^k :
  \sum_{i=1}^k \lambda_i = c, \lambda_i \geqslant 0\right\}.
\end{equation}

\subsection{Sequences of multivariate Bernoulli variables}

Consider now a sequence of independent and identically distributed multivariate
Bernoulli variables $\mathbf{B_1}, \mathbf{B_2}, \ldots, \mathbf{B_m} \sim 
Ber_k(\mathbf{p})$. The sum
\begin{equation}
  \mathbf{S}_m = \sum_{i=1}^m \mathbf{B}_i \sim Bi_k(m, \mathbf{p})
\end{equation}
is distributed as a \textit{multivariate Binomial random variable} \citep{krummenauer},
thus preserving one of the fundamental properties of the univariate Bernoulli
distribution. A similar result holds for the \textit{law of small numbers},
whose multivariate version states that a $k$-variate Binomial distribution
$Bi_k(m, \mathbf{p})$ converges to a \textit{multivariate Poisson distribution}
$P_k(\mathbf{\Lambda})$:
\begin{align}
  &\mathbf{S}_m \stackrel{d}{\to} P_k(\mathbf{\Lambda})&
  &\text{as}&
  &m\mathbf{p} \to \mathbf{\Lambda}.
\end{align}

Both these distributions' probability functions, while tractable, are not very
useful as a basis for closed-form inference procedures. An alternative is given
by the asymptotic \textit{multivariate Gaussian distribution} defined by the
\textit{multivariate central limit theorem} \citep{ash}:
\begin{equation}
\label{eqn:mclt}
  \frac{\mathbf{S}_m - m \E(\mathbf{B}_1)}{\sqrt{m}}
    \stackrel{d}{\to} N_k(\mathbf{0}, \Sigma).
\end{equation}
The limiting distribution is guaranteed to exist for all possible values of
$\mathbf{p}$, as the first two moments are bounded and therefore are always
finite.

\section{Inference on the network structure}

Let $\mathcal{U} = (\mathbf{V}, E)$ be the undirected graph underlying a DAG
$\mathcal{G} = (\mathbf{V}, A)$, defined as its unique biorientation
\citep{digraphs}. Each edge $e \in E$ of $\mathcal{U}$ corresponds to the directed
arcs in $A$ with the same incident nodes, and has only two possible states (it's
either present in or absent from the graph).

Then each possible edge $e_i$, $i = 1, \ldots, |\mathbf{V}|(|\mathbf{V}|-1)/2$ 
is naturally distributed as a Bernoulli random variable
\begin{equation}
  E_i = \left\{
    \begin{aligned}
     e_i &\in E&     &\text{with probability $p_i$}& \\
     e_i &\not\in E& &\text{with probability $1 - p_i$}&
    \end{aligned}
    \right.
\end{equation}
and every set $W \subseteq \mathbf{V} \times \mathbf{V}$ (including $E$) is
distributed as a multivariate Bernoulli random variable $\mathbf{W}$ and
identified by the parameter collection
$\mathbf{p}_W = \left\{ p_w : w \subseteq W, w \neq \varnothing \right\}$.
The elements of $\mathbf{p}_W$ can be estimated via parametric or nonparametric
bootstrap as in \citet{friedman}, because they are functions of the DAGs
$\mathcal{G}_b$, $b = 1, \ldots, m$ through the underlying undirected
graphs $\mathcal{U}_b = (V, E_b)$. The resulting empirical probabilities
\begin{equation}
  \hat p_w = \frac{1}{m} \sum_{b=1}^m \mathbb{I}_{\left\{w \subseteq E_b\right\}}(\mathcal{U}_b),
\end{equation}
in particular
\begin{align}
  &\hat p_i = \frac{1}{m} \sum_{b=1}^m \mathbb{I}_{\left\{e_i \in E_b\right\}}(\mathcal{U}_b)&
  &\text{and}&
  &\hat p_{ij} = \frac{1}{m} \sum_{b=1}^m \mathbb{I}_{\left\{e_i \in E_b, e_j \in E_b\right\}}(\mathcal{U}_b),
\end{align}
can be used to obtain several descriptive measures and test statistics for
the variability of the structure of a Bayesian network.

\subsection{Interpretation of bootstrapped networks}

Considering the undirected graphs $\mathcal{U}_1, \ldots, \mathcal{U}_m$
instead of the corresponding directed graphs $\mathcal{G}_1, \ldots,
\mathcal{G}_m$ greatly simplifies the interpretation of bootstrap's results. In
particular the variability of the graphical structure can be summarized in three
cases according to the entropy \citep{itheory} of the set of the bootstrapped
networks:
\begin{itemize}
  \item \textit{minimum entropy}: all the networks learned from the bootstrap
    samples have the same structure, that is $E_1 = E_2 = \ldots = E_m = E$.
    This is the best possible outcome of the simulation, because there is no
    variability in the estimated network. In this case the first two moments of
    the multivariate Bernoulli distribution are equal to
    \begin{align}
      &p_i = \left\{
        \begin{aligned}
         &1& &\text{if $e_i \in E$}     \\
         &0& &\text{otherwise}&
        \end{aligned}
        \right.&
      &\text{and}&
      &\Sigma = \mathbf{O}.
    \end{align}
  \item \textit{intermediate entropy}: several network structures are observed
    with different frequencies $m_b$, $\sum m_b = m$. The
    first two sample moments of the multivariate Bernoulli distribution are equal to
    \begin{align}
      &\hat p_i = \frac{1}{m} \sum_{b \,:\, e_i \in E_b} m_b&
      &\text{and}&
      &\hat p_{ij} = \frac{1}{m} \sum_{b \,:\, e_i \in E_b, e_j \in E_b} m_b.
    \end{align}
  \item \textit{maximum entropy}: all $2^{|\mathbf{V}|(|\mathbf{V}| - 1)/2}$ 
    possible network structures appear with the same frequency, that is
    \begin{align}
      &\hat\Prob(\mathcal{U}_i) = \frac{1}{2^{|\mathbf{V}|(|\mathbf{V}|-1)/2}}& &i = 1, \ldots, 2^{|\mathbf{V}|(|\mathbf{V}| - 1)/2}.
    \end{align}
    This is the worst possible outcome because edges vary independently of each
    other and each one is present in only half of the networks (proof provided in
    \ref{app:maxent}):
    \begin{align}
      \label{eqn:maxent}
      &p_i = \frac{1}{2}& &\text{and}& &\Sigma = \frac{1}{4} I_k.
    \end{align}
    This is also the only case in which all eigenvalues reach their maximum,
    that is $\lambda_1 = \lambda_2 = \ldots = \lambda_k = 1/4$.
\end{itemize}

\subsection{Descriptive statistics of network variability}
\label{sec:descriptive}

Several functions have been proposed in literature as univariate measures of
spread of a multivariate distribution, usually under the assumption of multivariate
normality (see for example \citet{muirhead} and \citet{bilodeau}). Three of
them in particular can be used as descriptive statistics for the multivariate
Bernoulli distribution:
\begin{itemize}
  \item the \textit{generalized variance}, $\VAR_G(\Sigma) = \det(\Sigma)$.
  \item the \textit{total variance}, $\VAR_T(\Sigma) = \tr(\Sigma)$, also called \textit{total variation} in \citet{mardia}.
  \item the squared \textit{Frobenius matrix norm},
    $\VAR_N(\Sigma) = ||| \Sigma - (k/4)I_k|||_F^2$.
\end{itemize}

Both generalized and total variance associate high values of the statistic
to unstable network structures, and are bounded due to the properties of the
covariance matrix $\Sigma$. For the total variance it's easy to show that
\begin{equation}
\label{eq:vart}
  0 \leqslant \VAR_T(\Sigma) = \sum_{i=1}^k \sigma_{ii} \leqslant \frac{1}{4}k.
\end{equation}
The generalized variance is similarly bounded due to Hadamard's theorem on the
determinant of a non-negative definite matrix \citep{seber}:
\begin{equation}
  0 \leqslant \VAR_G(\Sigma) \leqslant \prod_{i=1}^k \sigma_{ii} \leqslant \left(\frac{1}{4}\right)^k.
\end{equation}
They reach the respective maxima in the \textit{maximum entropy} case and
are equal to zero only in the \textit{minimum entropy} case. The generalized
variance is also strictly convex (the maximum is reached only for $\Sigma =
(1/4)I_k$), but it is equal to zero if $\Sigma$ is rank deficient.
For this reason it's convenient to reduce $\Sigma$ to a smaller, full rank matrix
(let's say $\Sigma^*$) and compute $\VAR_G(\Sigma^*)$ instead of $\VAR_G(\Sigma)$.

The squared difference in Frobenius norm between $\Sigma$ and $k$ times the
\textit{maximum entropy} covariance matrix associates high values of the
statistic to stable network structures. It can be rewritten in terms of the
eigenvalues $\lambda_1, \ldots, \lambda_k$ of $\Sigma$ as
\begin{equation}
  \VAR_N(\Sigma) = \sum_{i=1}^k \left( \lambda_i - \frac{k}{4}\right)^2.
\end{equation}
It has a unique maximum (in the \textit{minimum entropy} case), which can be
computed as the solution of the constrained minimization problem in
$\boldsymbol{\lambda} = [\lambda_1, \ldots, \lambda_k]^T$
\begin{align}
  &\min_{\mathcal{D}} f(\boldsymbol{\lambda})
    = -\sum_{i=1}^k \left( \lambda_i - \frac{k}{4}\right)^2&
  &\text{subject to}&
  &\lambda_i \geqslant 0, \sum_{i=1}^k \lambda_i \leqslant \frac{k}{4}
\end{align}
using Lagrange multipliers \citep{nocedal}. It also has a single minimum in
$\boldsymbol{\lambda}^* = [1/4, \ldots, 1/4]$, which is the projection of 
$[k/4, \ldots, k/4]$ onto the set $\mathcal{D}$ and coincides with the 
\textit{maximum entropy} case. The proof for these boundaries and the rationale
behind the use of $(k/4)I_k$ instead of $(1/4)I_k$ are reported in \ref{app:frob}.

The corresponding normalized statistics are:
\begin{align}
  \overline{\VAR}_T(\Sigma) &= \frac{\VAR_T(\Sigma)}{\max_{\Sigma} \VAR_T(\Sigma)} = \frac{4 \VAR_T(\Sigma)}{k} \\
  \overline{\VAR}_G(\Sigma) &= \frac{\VAR_G(\Sigma)}{\max_{\Sigma} \VAR_G(\Sigma)} = 4^k \VAR_G(\Sigma) \\
  \overline{\VAR}_N(\Sigma) &= \frac{\max_{\Sigma} \VAR_N(\Sigma) - \VAR_N(\Sigma)}{\max_{\Sigma} \VAR_N(\Sigma) - \min_{\Sigma} \VAR_N(\Sigma)}
    = \frac{k^3 - 16\VAR_N(\Sigma)}{k(2k - 1)}.
\end{align}
All of them vary in the $[0,1]$ interval and associate high values of the
statistic to networks whose structure display a high variability across the
bootstrap samples. Equivalently we can define their complements 
$\overline{\overline{\VAR}}_T(\Sigma)$, $\overline{\overline{\VAR}}_G(\Sigma)$
and $\overline{\overline{\VAR}}_N(\Sigma)$, which associate high values of 
the statistic to networks with little variability and can be used as 
measures of distance from the \textit{maximum entropy} case.

\subsection{Asymptotic inference}
\label{sec:aymptotic}

The limiting distribution of the descriptive statistics defined above can be
derived by replacing the covariance matrix $\Sigma$ with its unbiased estimator
$\hat\Sigma$ and by considering the multivariate Gaussian distribution from equation
\ref{eqn:mclt}. The hypothesis we are interested in is
\begin{align}
\label{eqn:null}
  &H_0: \Sigma = \frac{1}{4}I_k&
  &H_1: \Sigma \neq \frac{1}{4}I_k,
\end{align}
which relates the sample covariance matrix with the one from the \textit{maximum
entropy} case.

For the total variance we have that $t_T = 4m \tr(\hat\Sigma) \stackrel{.}{\sim} \chi^2_{mk}$
\citep{muirhead},
and since the maximum value of $\tr(\Sigma)$ is achieved in the \textit{maximum
entropy} case, the hypothesis in Equation \ref{eqn:null} assumes the form
\begin{align}
  &H_0: \tr(\Sigma) = \frac{k}{4}&
  &H_1: \tr(\Sigma) < \frac{k}{4}.
\end{align}
Then the observed significance value is $\hat\alpha_T = \Prob(t_T \leqslant t_T^{oss})$,
and can be improved with the finite sample correction
\begin{equation}
   \tilde\alpha_T = \Prob\left(t_T \leqslant t_T^{oss} \,|\, t_T \in [0, mk]\right) =
   \frac{\Prob(t_T \leqslant t_T^{oss})}{\Prob(t_T \leqslant mk)}
\end{equation}
which accounts for the bounds on $\VAR_T(\Sigma)$ from inequality \ref{eq:vart}.

For the generalized variance there are several possible asymptotic and approximate
distributions:
\begin{itemize}
  \item the Gaussian distribution defined in \citet{anderson}
    \begin{equation}
      t_{G_1} = \sqrt{m} \left( \frac{\det(\hat\Sigma)}{\det(\frac{1}{4}I_k)} -1 \right) \stackrel{.}{\sim} N(0, 2k).
    \end{equation}
  \item the Gamma distribution defined in \citet{steyn}
    \begin{equation}
      t_{G_2} = \frac{mk}{2} \sqrt[k]{\frac{\det(\hat\Sigma)}{\det(\frac{1}{4}I_k)}} \stackrel{.}{\sim} Ga\left(\frac{k(m+1-k)}{2}, 1\right).
    \end{equation}
  \item the saddlepoint approximation defined in \citet{saddlepoint}.
\end{itemize}

As before the hypothesis in Equation \ref{eqn:null} assumes the form
\begin{align}
  &H_0: \det(\Sigma) = \det\left(\frac{1}{4}I_k\right)&
  &H_1: \det(\Sigma) < \det\left(\frac{1}{4}I_k\right).
\end{align}
The observed significance values for the Gaussian and Gamma distributions are
$\hat\alpha_{G_1} = \Prob(t_{G_1} \leqslant t_{G_1}^{oss})$ and
$\hat\alpha_{G_2} = \Prob(t_{G_2} \leqslant t_{G_2}^{oss})$, and
the respective finite sample corrections for the bounds on $\det(\Sigma)$ are
\begin{align}
   \tilde\alpha_{G_1} &= \Prob\left(t_{G_1} \leqslant t_{G_1}^{oss} \,|\, t_{G_1} \in \left[-\sqrt{m}, 0\right]\right)
     = \frac{\Prob(t_{G_1} \leqslant t_{G_1}^{oss}) - \Prob(t_{G_1} \leqslant -\sqrt{m})}{\Prob(t_{G_1} \leqslant 0) -  \Prob(t_{G_1} \leqslant -\sqrt{m})}\\
   \tilde\alpha_{G_2} &= \Prob\left(t_{G_2} \leqslant t_{G_2}^{oss} \,|\, t_{G_2} \in \left[0, \frac{mk}{2}\right]\right)
     =  \frac{\Prob(t_{G_2} \leqslant t_{G_2}^{oss})}{\Prob(t_{G_2} \leqslant \frac{mk}{2})}.
\end{align}

The test statistic associated with the squared Frobenius norm is the test for
the equality of two covariance matrices defined in \citet{nagao},
\begin{equation}
  t_N = \frac{m}{2} \tr\left(\left[\hat\Sigma \left(\frac{1}{4}I_k\right)^{-1} - I_k\right]^2\right)
    = \frac{m}{2} \tr\left(\left[4\hat\Sigma  - I_k\right]^2\right) \stackrel{.}{\sim} \chi^2_{\frac{1}{2}k(k+1)},
\end{equation}
because
\begin{equation}
  \tr\left(\left[4\hat\Sigma - I_k\right]^2\right) 
    = 16 \sum_{i=1}^k \left(\lambda_i - \frac{1}{4}\right)^2 = 16 ||| \hat\Sigma - \frac{1}{4}I_k|||_F^2
\end{equation}
See \ref{app:frob} for an explanation of the use of $(1/4)I_k$ instead of $(k/4)I_k$.
The significance value for $t_N$ is $\hat\alpha_N = \Prob(t_N \geqslant t_N^{oss})$
as the hypothesis in Equation \ref{eqn:null} becomes
\begin{align}
  &H_0: ||| \Sigma - \frac{1}{4}I_k|||_F = 0&
  &H_1: ||| \Sigma - \frac{1}{4}I_k|||_F > 0.
\end{align}
Unlike the previous statistics, Nagao's test displays a good convergence 
speed, to the point that the finite sample correction for the bounds on the 
squared Frobenius matrix norm
\begin{equation}
   \tilde\alpha_N = \Prob\left(t_N \geqslant t_N^{oss} \,|\, t_{G_1} \in \left[0, t_N^{max}\right]\right)
     = \frac{\Prob(t_N \geqslant t_N^{oss}) - \Prob(t_N > t_N^{max})}{\Prob(t_N \leqslant t_N^{max})}
\end{equation}
is not appreciably better than the raw significance value (see Table \ref{tab:asym}
for a simple example).

\subsection{Monte Carlo inference and parametric bootstrap}

Another approach to compute the significance values associated with $\VAR_T(\Sigma)$, 
$\VAR_G(\Sigma)$ and $\VAR_N(\Sigma)$ is applying again parametric bootstrap.

The multivariate Bernoulli distribution $\mathbf{W}_0$ specified by the hypothesis
in \ref{eqn:null} has a diagonal covariance matrix, so its components $W_{0_i}$,
$i = 1, \ldots, k$ are uncorrelated. According to Theorem \ref{thm:univindep} they
are also independent, so the joint distribution of $\mathbf{W}_0$ is completely
specified by the marginal distributions $W_{0_i} \sim Ber(1/2)$. Therefore
it's possible (and indeed quite easy) to generate observations from the null
distribution and use them to estimate the significance value of the normalized
statistics $\overline{\overline{\VAR}}_T(\Sigma)$, $\overline{\overline{\VAR}}_G(\Sigma)$
and $\overline{\overline{\VAR}}_N(\Sigma)$ defined in section \ref{sec:descriptive}:

\begin{enumerate}
  \item compute the value of test statistic $T$ on the original covariance matrix $\Sigma$.
  \item For $r = 1, 2, \ldots, R$.
  \begin{enumerate}
    \item generate $m$ sets of $k$ random samples from a $Ber(1/2)$ distribution.
    \item compute their covariance matrix $\Sigma^*_r$.
    \item compute $T^*_r$ from $\Sigma^*_r$
  \end{enumerate}
  \item compute the Monte Carlo significance value as 
    $\hat\alpha_R = (1/R) \sum_{r = 1}^R \mathbb{I}_{\{x \geqslant T\}}(T^*_r)$.
\end{enumerate}

This approach has two important advantages over the parametric tests defined in section
\ref{sec:aymptotic}:
\begin{itemize}
  \item the test statistic is evaluated against its true null distribution instead
    of its asymptotic approximation, thus removing any distortion caused by lack
    of convergence (which can be quite slow and problematic in high dimensions).
  \item each simulation $r$ has a lower computational cost than the equivalent
    application of the structure learning algorithm to a bootstrap sample $b$.
    Therefore the Monte Carlo test can achieve a good precision with a smaller
    number of bootstrapped networks, allowing its application to larger problems.
\end{itemize}

\section{A simple example}

\begin{figure}[ht]
  \begin{center}
    \includegraphics[width=10cm]{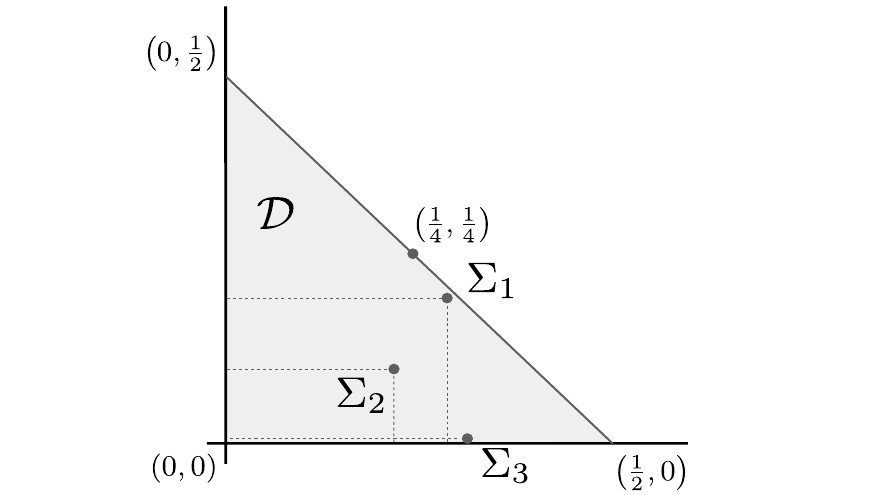}
  \end{center}
  \caption{The covariance matrices $\Sigma_1$, $\Sigma_2$ and $\Sigma_3$ represented
    as functions of their eigenvalues in $\mathcal{D}$ (grey). The points $(0,0)$
    and $(1/4, 1/4)$ correspond to the \textit{minimum entropy} and
    \textit{maximum entropy} cases.}
\label{fig:base}
\end{figure}

Consider the multivariate Bernoulli distributions $\mathbf{W}_1$, $\mathbf{W}_2$ and
$\mathbf{W}_3$ with second order moments
\begin{align}
  &\Sigma_1 = \frac{1}{25} \begin{bmatrix} 6 & 1 \\ 1 & 6 \end{bmatrix},&
  &\Sigma_2 = \frac{1}{625} \begin{bmatrix} 66 & -21 \\ -21 & 126 \end{bmatrix},&
  &\text{and}&
  &\Sigma_3 = \frac{1}{625} \begin{bmatrix} 66 & 91 \\ 91 & 126 \end{bmatrix}
\end{align}
associated with two (increasingly correlated) arcs from networks. The 
eigenvalues of $\Sigma_1$, $\Sigma_2$ and $\Sigma_3$ are
\begin{align}
  &\boldsymbol{\lambda}_1 = \begin{bmatrix} 0.28 \\ 0.20 \end{bmatrix},&
  &\boldsymbol{\lambda}_2 = \begin{bmatrix} 0.2121 \\ 0.095 \end{bmatrix}&
  &\text{and}&
  &\boldsymbol{\lambda}_3 = \begin{bmatrix} 0.3069 \\ 0.0003 \end{bmatrix}.
\end{align}
The values of the generalized variance, total variance and squared Frobenius
matrix norm (both normalized and in the original scale) for the three covariance
matrices are reported int Table \ref{tab:num}. 
\begin{table}[!b]
  \begin{center}
  \begin{tabular}{|l|lll|lll|}
  \hline
             &                  &                  &                  &                             &                             & \\[-12pt]
             & $\VAR_T(\Sigma)$ & $\VAR_G(\Sigma)$ & $\VAR_N(\Sigma)$ & $\overline{\VAR}_T(\Sigma)$ & $\overline{\VAR}_G(\Sigma)$ & $\overline{\VAR}_N(\Sigma)$ \\
  \hline
  $\Sigma_1$ & $0.48$           & $0.056$          & $0.1384$         & $ 0.96$                     & $0.896$                     & $0.9642$  \\
  $\Sigma_2$ & $0.3072$         & $0.02016$        & $0.2468$         & $0.6144$                    & $0.32256$                   & $0.6752$     \\
  $\Sigma_3$ & $0.3072$     & $8.96\times 10^{-5}$ & $0.2869$         & $0.6144$                    & $0.00143$                   & $0.5682$  \\
  \hline
  \end{tabular}
  \end{center}
  \caption{Original and normalized values of $\VAR_T$, $\VAR_G$ and 
    $\VAR_N$ for $\Sigma_1$, $\Sigma_2$ and $\Sigma_3$.}
  \label{tab:num}
\end{table}

The corresponding asymptotic and Monte Carlo significance values are reported in
Table \ref{tab:asym} and \ref{tab:boot} respectively. Each one has been computed
for various hypothetical sample sizes ($m = 10, 20, 50, 100, 200$). Parametric
bootstrap has been performed on $R = 10^6$ covariance matrices generated from the
null distribution for each configuration of test statistic and sample size.

\begin{table}[hp!]
  \begin{center}
    \begin{tabular}{|l|lllll|}
    \hline
    \multicolumn{6}{|c|}{$t_T(\Sigma)$} \\
    \hline
               & $10$        & $20$        & $50$        & $100$          & $200$      \\
    \hline
    \multirow{2}{*}{$\Sigma_1$}
               & $0.491137$ & $0.457610$ & $0.405404$ & $0.354943$    & $0.291243$ \\
               & $\mathbf{0.906041}$ & $\mathbf{0.863836}$ & $\mathbf{0.781414}$ & $\mathbf{0.691495}$ & $\mathbf{0.571734}$  \\
    \multirow{2}{*}{$\Sigma_2$}
               & $0.094193$ & $0.026330$ & $0.000852$ & $0.000003$    & $0.000000$ \\
               & $\mathbf{0.173766}$ & $\mathbf{0.049704}$ & $\mathbf{0.001644}$ & $\mathbf{0.000007}$ & $\mathbf{0.000000}$ \\
    \multirow{2}{*}{$\Sigma_3$}
               & $0.094193$ & $0.026330$ & $0.000852$ & $0.000003$    & $0.000000$ \\
               & $\mathbf{0.173766}$ & $\mathbf{0.049704}$ & $\mathbf{0.001644}$ & $\mathbf{0.000007}$ & $\mathbf{0.000000}$ \\
    \hline
    \multicolumn{6}{|c|}{$t_{G_2}(\Sigma)$} \\
    \hline
    \multirow{2}{*}{$\Sigma_1$}
               & $0.603944$ & $0.524258$ & $0.423183$     & $0.341131$     & $0.250054$ \\
               & $\mathbf{0.905218}$ & $\mathbf{0.847522}$ & $\mathbf{0.735799}$ & $\mathbf{0.616696}$ & $\mathbf{0.465129}$ \\
    \multirow{2}{*}{$\Sigma_2$}
               & $0.121488$ & $0.023514$ & $0.000278$    & $0.000000$     & $0.000000$ \\
               & $\mathbf{0.182091}$ & $\mathbf{0.0380138}$ & $\mathbf{0.000484}$ & $\mathbf{0.000000}$ & $\mathbf{0.000000}$ \\
    \multirow{2}{*}{$\Sigma_3$}
               & $0.000000$ & $0.000000$ & $0.000000$ & $0.000000$ & $0.000000$ \\
               & $\mathbf{0.000000}$ & $\mathbf{0.000000} $ & $ \mathbf{0.000000} $ & $\mathbf{0.000000}$ & $\mathbf{0.000000} $\\
    \hline
    \multicolumn{6}{|c|}{$t_N(\Sigma)$} \\
    \hline
    \multirow{2}{*}{$\Sigma_1$}
               & $0.965205$ & $0.909123$ & $0.714937$     & $0.436839$    & $0.142271$ \\
               & $\mathbf{0.964547}$ & $\mathbf{0.909108}$ & $\mathbf{0.714937}$ &  $\mathbf{0.436839}$ & $\mathbf{0.142271}$ \\
    \multirow{2}{*}{$\Sigma_2$}
               & $0.564938$ & $0.253762$ & $0.017090$     & $0.000142$    & $0.000000$ \\
               & $\mathbf{0.556708}$ & $\mathbf{0.253636}$ & $\mathbf{0.017090}$ & $\mathbf{0.000142}$ & $\mathbf{0.000000}$ \\
    \multirow{2}{*}{$\Sigma_3$}
               & $0.154551$ & $0.014796$ & $0.000008$     & $0.000000$ & $0.000000$ \\
               & $\mathbf{0.138557}$ & $\mathbf{0.014628}$ & $\mathbf{0.000008}$ & $\mathbf{0.000000}$ & $\mathbf{0.000000}$ \\
    \hline
    \end{tabular}
  \caption{Asymptotic significance values of $t_T$, $t_{G_2}$ and $t_N$
    for $\Sigma_1$, $\Sigma_2$ and $\Sigma_3$; the ones computed with the
    finite sample corrections are reported in bold.}
  \label{tab:asym}
  \begin{tabular}{|l|lllll|}
  \hline
  \multicolumn{6}{|c|}{} \\[-12pt]
  \multicolumn{6}{|c|}{$\overline{\overline{\VAR}}_T(\Sigma)$} \\
  \hline
             & $10$       & $20$       & $50$       & $100$      & $200$      \\
  \hline
  $\Sigma_1$ & $0.569655$ & $0.457109$ & $0.129242$ & $0.017416$ & $0.000334$ \\
  $\Sigma_2$ & $0.016834$ & $0.000205$ & $0$        & $0$        & $0$        \\
  $\Sigma_3$ & $0.016834$ & $0.000205$ & $0$        & $0$        & $0$        \\
  \hline
  \multicolumn{6}{|c|}{} \\[-12pt]
  \multicolumn{6}{|c|}{$\overline{\overline{\VAR}}_G(\Sigma)$} \\
  \hline
  $\Sigma_1$ & $0.784102$ & $0.512839$ & $0.14788$  & $0.013678$ & $0.000094$ \\
  $\Sigma_2$ & $0.063548$ & $0.000761$ & $0$        & $0$        & $0$        \\
  $\Sigma_3$ & $0.005909$ & $0.000008$ & $0$        & $0$        & $0$        \\
  \hline
  \multicolumn{6}{|c|}{} \\[-12pt]
  \multicolumn{6}{|c|}{$\overline{\overline{\VAR}}_N(\Sigma)$} \\
  \hline
  $\Sigma_1$ & $0.743797$ & $0.568819$ & $0.239397$ & $0.096544$ & $0.019633$ \\
  $\Sigma_2$ & $0.196996$ & $0.037772$ & $0.001018$ & $0.000005$ & $0$        \\
  $\Sigma_3$ & $0.018292$ & $0.000355$ & $0$        & $0$        & $0$        \\
  \hline
  \end{tabular}
  \caption{Bootstrap significance values of $\overline{\overline{\VAR}}_T$, 
    $\overline{\overline{\VAR}}_G$ and $\overline{\overline{\VAR}}_N$ from 
    parametric bootstrap for $\Sigma_1$, $\Sigma_2$ and $\Sigma_3$.}
  \label{tab:boot}
\end{center}
\end{table}

\pagebreak

\section{Comparing independence tests and structure learning algorithms}

We will now illustrate how these tests can be used to compare different structure
learning strategies, i.e. different combinations of structure learning algorithms,
conditional independence tests and network scores. The impact of different choices
for each component on the variability of the model can easily be assessed while
keeping the other ones fixed. 

\begin{figure}[!b]
  \begin{center}
    \includegraphics[width=\textwidth]{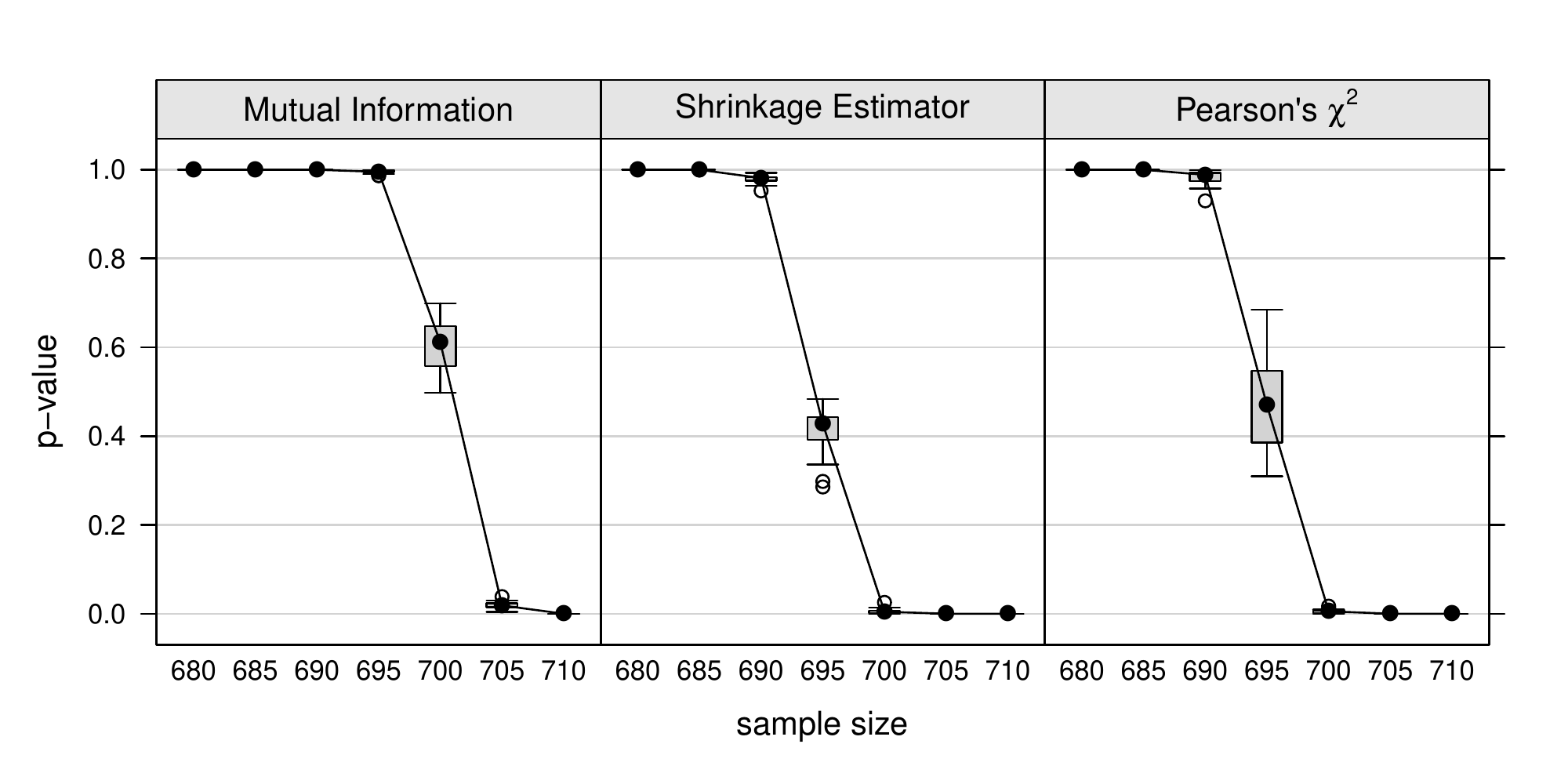}
    \caption{Significance values for three different conditional independence tests
    (asymptotic and shrinkage estimators of mutual information and Pearson's $\chi^2$)
    used with the same structure learning algorithm (Grow-Shrink).}
    \label{fig:smalltest}
  \end{center}
\end{figure}

First we will compare the performance of the Grow-Shrink algorithm for three 
different conditional independence tests. The learning algorithm has been applied
to samples of size $680$, $685$, $690$, $695$, $700$, $705$ and $710$ (20 for 
each size) generated from the ALARM reference network \citep{alarm}, which is 
composed by 37 discrete nodes and 46 arcs for a total of 509 parameters. Both the
data and the software implementation of the algorithm are included in the bnlearn
package \citep{bnlearn} for R \citep{r}. The following tests have been considered: 
\begin{itemize}
  \item the asymptotic $\chi^2$ test based on mutual information \citep{itheory}, 
    which is in fact a log-likelihood ratio test and is also called the $G^2$ 
    test \citep{agresti}.
  \item the shrinkage estimator for the mutual information, which is a James-Stein
    regularized estimator developed by \citet{shrinkage}.
  \item Pearson's $\chi^2$ asymptotic test for independence \citep{agresti}. 
\end{itemize}
The same threshold $\alpha = 0.05$ for type I error has been used in three cases,
and network variability has been assessed with the Monte Carlo test for the
squared Frobenius norm.

Results are shown in Figure \ref{fig:smalltest}. All the tests considered in the 
analysis start producing relatively stable network structures -- i.e. the null
hypothesis corresponding to the maximum entropy case is rejected --  at sample 
sizes $695$ and $700$. Pearson's $\chi^2$ test performs slightly better than
mutual information, as documented in \citet{agresti} when dealing with sparse
contingency tables. This is also true for the shrinkage estimator. However, 
the difference among the three sets of significance values is very small.

On the other hand we will now compare three different learning algorithms:
\begin{itemize}
  \item TABU search (which is a score-based algorithm), combined with a
    Bayesian Information criterion (BIC) score.
  \item Grow-Shrink (which is a constraint-based algorithm) combined with the 
    asymptotic $\chi^2$ test based on mutual information described above and
    $\alpha = 0.05$.
  \item Max-Min Hill Climbing (which is hybrid algorithm), combined with
    a BIC score and the asymptotic mutual information test.
\end{itemize}
As can be seen in Figure \ref{fig:smallalgo} in this case differences are
more pronounced. The Max-Min Hill Climbing algorithm, which is one of the top
performers up to date for large networks, displays less variability than 
TABU search and Grow-Shrink at the same sample size. In particular the 
difference between Max-Min Hill Climbing and Grow-Shrink confirms the analysis
made in \citet{mmhc} and the well-documented \citep{spirtes} instability 
displayed by constraint-based algorithms at small sample sizes.

\begin{figure}[!h]
  \begin{center}
    \includegraphics[width=\textwidth]{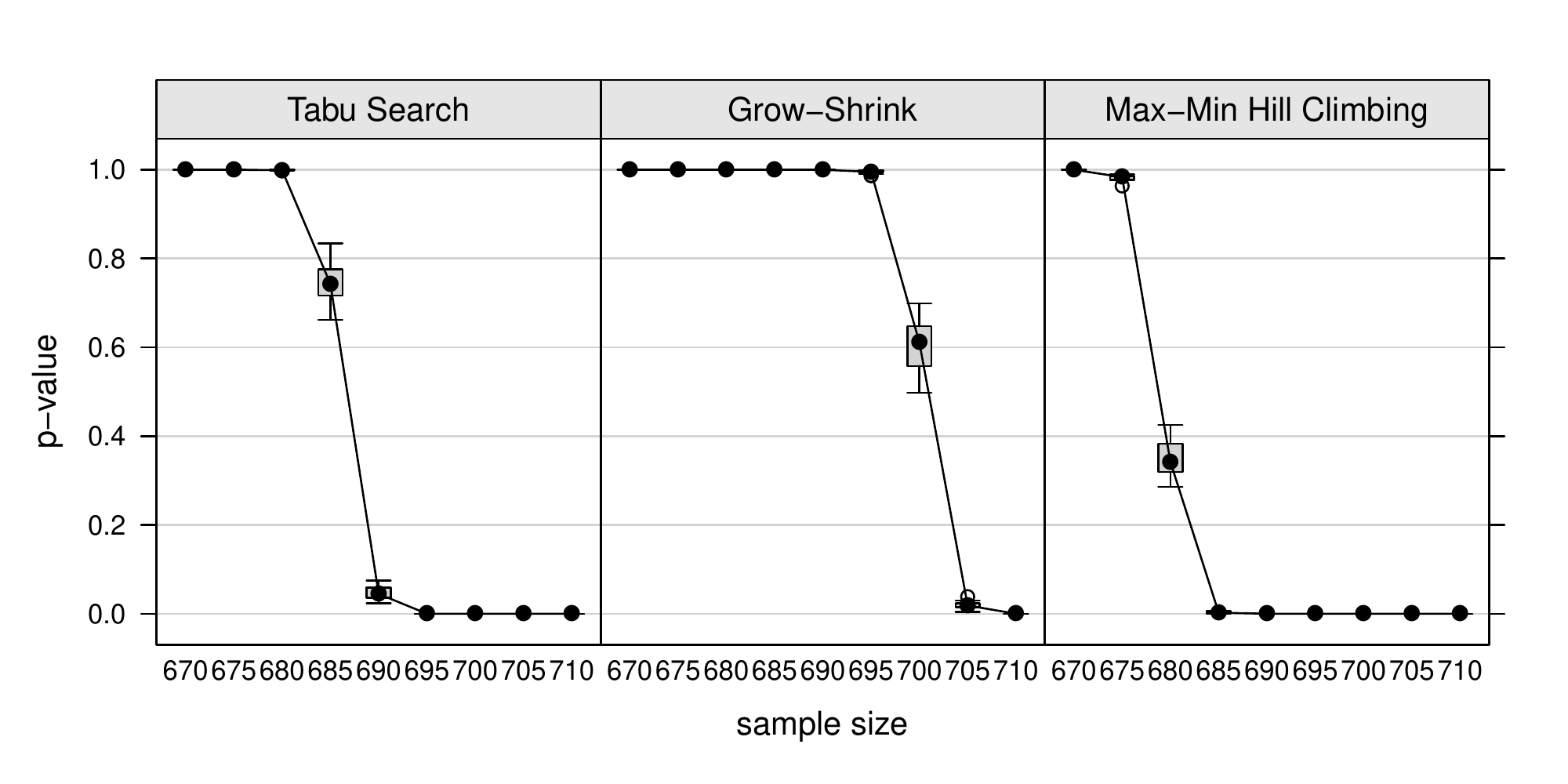}
    \caption{Significance values for three different structure learning algorithms
    (Grow-Shrink, TABU search and Max-Min Hill Climbing) using the same conditional
    independence tests and network scores (asymptotic mutual information test and
    the Bayesian Information Criterion (BIC), respectively).}
    \label{fig:smallalgo}
  \end{center}
\end{figure}

\section{Conclusions}

In this paper we derived the properties of several measures of variability for
the structure of a Bayesian network through its underlying undirected graph,
which is assumed to have a multivariate Bernoulli distribution. Descriptive
statistics, asymptotic and Monte Carlo tests were developed along with their
fundamental properties. They can be used to compare the performance of different
learning algorithms and to measure the strength of arbitrary subsets of arcs.

\section*{Acknowledgements}

Many thanks to Prof. Adriana Brogini, my Supervisor at the Ph.D. School in 
Statistical Sciences (University of Padova), for proofreading this article and
giving many useful comments and suggestions. I would also like to thank Giovanni
Andreatta and Luigi Salce (Full Professors at the Department of Pure and Applied
Mathematics, University of Padova) for their help in the development of the
constrained optimization and matrix norm applications respectively.

\section*{Appendix}

\appendix

\section{Bounds on the squared Frobenius matrix norm}
\label{app:frob}

The squared Frobenius matrix norm of the difference between the covariance
matrix $\Sigma$ and the \textit{maximum entropy} matrix $(1/4)I_k$ is
\begin{equation}
   |||\Sigma - \frac{1}{4}I_k|||_F^2
    = \sum_{i=1}^k \left( \lambda_i - \frac{1}{4}\right)^2.
\end{equation}
Its unique global minimum is zero for $\Sigma = (1/4)I_k$ but it has a varying
number of global maxima depending on the dimension $k$ of $\Sigma$. They are
the solutions of the constrained minimization problem
\begin{align}
  &\min_{\mathcal{D}} f(\boldsymbol{\lambda})
    = -\sum_{i=1}^k \left( \lambda_i - \frac{k}{4}\right)^2&
  &\text{subject to}&
  &\lambda_i \geqslant 0, \sum_{i=1}^k \lambda_i \leqslant \frac{k}{4}.
\end{align}
This configuration of stationary points is not a problem for asymptotic and Monte
Carlo tests, but prevents any direct interpretation of the values of descriptive
statistics.
\begin{figure}[t]
  \begin{center}
    \includegraphics{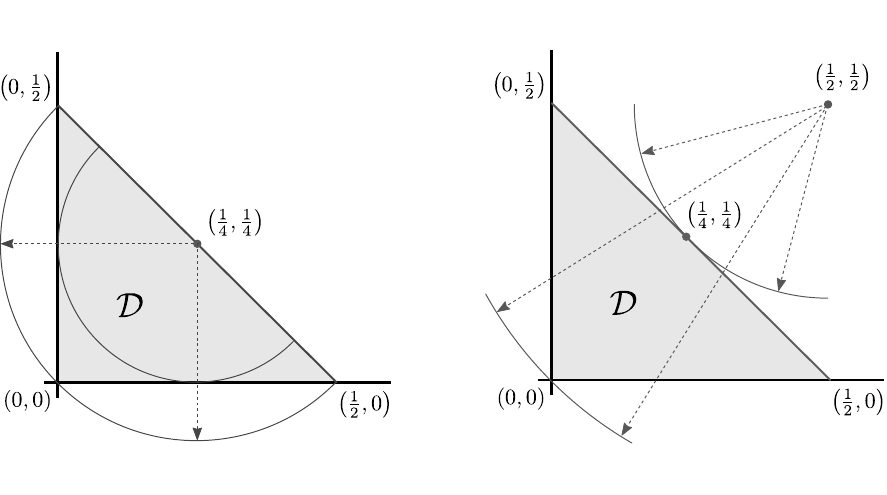}
  \end{center}
  \caption{Squared Frobenius matrix norms from $(1/4)I_K$ (on the left) and
    $(k/4)I_k$ (on the right) in $\mathcal{D}$ for $k = 2$. The green area
    is the set $\mathcal{D}$ of the possible eigenvalues of $\Sigma$ and the red
    lines are level curves.}
\label{fig:frobenius}
\end{figure}

On the other hand, the difference in squared Frobenius norm
\begin{equation}
  \VAR_N(\Sigma) = |||\Sigma - \frac{k}{4}I_k|||_F^2
    = \sum_{i=1}^k \left( \lambda_i - \frac{k}{4}\right)^2
\end{equation}
has both a unique global minimum (because it's a strictly convex function)
\begin{equation}
  \min_{\mathcal{D}} \VAR_N(\Sigma) = \VAR_N\left(\frac{1}{4}I_k\right)
    = \sum_{i=1}^k \left( \frac{1}{4} - \frac{k}{4}\right)^2 = \frac{k(k-1)^2}{16}
\end{equation}
and a unique global maximum
\begin{equation}
  \max_{\mathcal{D}} \VAR_N(\Sigma) = \VAR_N(\mathbf{O})
     = \sum_{i=1}^k \left( \frac{k}{4}\right)^2 = \frac{k^3}{16}
\end{equation}
which correspond to the \textit{minimum entropy} ($\boldsymbol{\lambda} = [0, \ldots, 0]$)
and the \textit{maximum entropy} ($\boldsymbol{\lambda} = [1/4, \ldots, 1/4]$)
covariance matrices respectively (see figure \ref{fig:frobenius}). However since
$(k/4)I_k$ is not a valid covariance matrix for a multivariate Bernoulli
distribution, $\VAR_N(\Sigma)$ cannot be used to derive any probabilistic result.

\section{Multivariate Bernoulli and the maximum entropy case}
\label{app:maxent}

The values of $p_i$ and $\Sigma$ in the \textit{maximum entropy} case are a direct
consequence from the following theorem.
\begin{thm}
  Let $\mathcal{U}_1, \ldots, \mathcal{U}_n$, $n = 2^m$, $m = |\mathbf{V}| (|\mathbf{V}| - 1)/2$
  be all possible undirected graphs with vertex set $\mathbf{V}$ and let
  $\Prob(\mathcal{U}_k) = 1/n$, $k = 1, \ldots, n$. Let $e_i$ and $e_j$,
  $i \neq j$ be two edges. Then $\Prob(e_i) = 1/2$ and $\Prob(e_i, e_j) = 1/4$.
\end{thm}
\begin{proof}
  The number of possible configurations of an undirected graph is given by
  the Cartesian product of the possible states of its $m$ edges, resulting in
  \begin{equation}
    |\{0,1\} \times \ldots \times \{0,1\}| = \left|\{0,1\}^m\right| = 2^m
  \end{equation}
  possible undirected graphs. Then edge $e_i$ is present in $2^{m-1}$ graphs
  and $e_i$ and $e_j$  are simultaneously present in $2^{m-2}$ graphs. Therefore
  \begin{align}
    &\Prob(e_i) = \frac{2^{m-1} \Prob(\mathcal{U}_k) }{2^m \Prob(\mathcal{U}_k) } = \frac{1}{2}&
    &\text{and}&
    &\Prob(e_i, e_j) = \frac{2^{m-2} \Prob(\mathcal{U}_k) }{2^m \Prob(\mathcal{U}_k) } = \frac{1}{4}.
  \end{align}
\end{proof}

The fact that $\sigma_{ij} = 0$ for every $i \neq j$ also proves
that the edges are independent according to Theorem \ref{thm:univindep}.

\end{document}